\documentclass[a4paper,UKenglish,cleveref,thm-restate]{lipics-v2021}

\usepackage{nicefrac}
\usepackage{xspace}
\usepackage{complexity}

\usepackage[subrefformat=simple,labelformat=simple]{subcaption}

\crefname{claim}{Claim}{Claims}

\hideLIPIcs

\graphicspath{{./figures/}}

\title{The Dispersive Art Gallery Problem}

\titlerunning{The Dispersive Art Gallery Problem}

\author{Christian Rieck}{Department of Computer Science, TU Braunschweig, Braunschweig, Germany}{rieck@ibr.cs.tu-bs.de}{https://orcid.org/0000-0003-0846-5163}{}
\author{Christian Scheffer}{Faculty of Electrical Engineering and Computer Science, Bochum University of Applied Sciences, Bochum, Germany}{christian.scheffer@hs-bochum.de}{https://orcid.org/0000-0002-3471-2706}{}

\authorrunning{C. Rieck and C. Scheffer}

\Copyright{Christian Rieck and Christian Scheffer}



\acknowledgements{We thank Joseph S. B. Mitchell for bringing this problem to our attention.}

\nolinenumbers

\EventEditors{}
\EventNoEds{0}
\EventLongTitle{International Symposium on Algorithms and Computation}
\EventShortTitle{ISAAC 2022}
\EventAcronym{ISAAC}
\EventYear{2022}
\EventDate{}
\EventLocation{}
\EventLogo{}
\SeriesVolume{}
\ArticleNo{73}

\newcommand{\guardset}{\mathcal{G}}
\newcommand{\polygon}{\mathcal{P}}
\newcommand{\visregion}{\mathcal{V}}

\begin{document}

\maketitle

\begin{abstract}
	We introduce a new variant of the art gallery problem that comes from safety issues. 
	In this variant we are not interested in guard sets of smallest cardinality, but in guard sets with largest possible distances between these guards. 
	To the best of our knowledge, this variant has not been considered before.
	We call it the \textsc{Dispersive Art Gallery Problem}. 
	In particular, in the dispersive art gallery problem we are given a polygon $\polygon$ and a real number $\ell$, and want to decide whether $\polygon$ has a guard set such that every pair of guards in this set is at least a distance of $\ell$ apart.

	In this paper, we study the vertex guard variant of this problem for the class of polyominoes. We consider rectangular visibility and distances as geodesics in the $L_1$-metric.
	Our results are as follows.
	We give a (simple) thin polyomino such that every guard set has minimum pairwise distances of at most $3$.
	On the positive side, we describe an algorithm that computes guard sets for simple polyominoes that match this upper bound, i.e., the algorithm constructs worst-case optimal solutions.
	We also study the computational complexity of computing guard sets that maximize the smallest distance between all pairs of guards within the guard sets. 
	We prove that deciding whether there exists a guard set realizing a minimum pairwise distance for all pairs of guards of at least $5$ in a given polyomino is \NP-complete.

	We were also able to find an optimal dynamic programming approach that computes a guard set that maximizes the minimum pairwise distance between guards in tree-shaped polyominoes, i.e., computes optimal solutions.
	Because the shapes constructed in the \NP-hardness reduction are thin as well (but have holes), this result completes the case for thin polyominoes.
\end{abstract}

\section{Introduction}\label{sec:introduction}

How many guards are necessary to guard an art gallery? 
This question was first posed by Victor Klee in 1973 and opened a flourishing field of research in computational geometry; see~for example the book by O'Rourke~\cite{o1987art}, or the surveys by  Shermer~\cite{shermer1992recent}, and Urrutia~\cite{Urrutia00}. 
This question states the classic \textsc{Art~Gallery~Problem} as follows: 
Given a (simple) polygon~$\polygon$ and an integer $k$, decide whether there is a guard set of cardinality $k$ such that every point~$p\in \polygon$ is seen by at least one guard, where a point is seen by a guard if and only if the connecting line segment is inside the polygon.

Suppose the following situation: Your art gallery is the victim of a robbery, or there is a fire outbreak and heavy smoke development in one part of the building.
Because guards in an optimal solution to instances of the classic art gallery problem can be really close together, many cameras can be affected at the same time, see~\cref{fig:introduction}.
From safety and security issues this would be a catastrophic scenario.
We want to address these issues, i.e., for a given shape, we are interested in a guard set that realizes preferably large distances between any two guards of the respective set, rather than focusing on the minimum number of guards needed.
Problems of this kind are called \textsc{Dispersion Problems}, and are typically stated as follows:
Given a set of $n$ objects in the plane and an integer $k$, decide if there is a subset of $k$ such objects, such that the distances between any pair in this subset is at least as large as a given threshold.
We assume that the shortest paths that realize the distances between guards are within the shape, i.e., they do not leave and enter the shape.

In this paper, we introduce the following problem that combines art gallery and dispersion problems and is described as follows.

\begin{description}
	\item[Dispersive Art Gallery Problem] Given a polygon $\polygon$ and a real number $\ell$, decide whether there exists a guard set~$\guardset$ for $\polygon$ such that the pairwise geodesic distances between any two guards in $\guardset$ are at least $\ell$.
\end{description}

Note that in this problem we are not interested in the size of a particular guard set, but only in the distances between guards realized by the guard set.  To the best of our knowledge, this problem has not been considered before. 
Additionally, a first intuitive thought might be that solutions to the classic art gallery problem are also solutions to this variant, since small cardinality guard sets should somehow yield larger pairwise distances. 
However, this is nowhere near the truth, see for example~\cref{fig:introduction} where doubling the size of the guard set results in an arbitrary growth of the dispersion distance.

\begin{figure}[h]
	\centering
	\includegraphics[scale=.75]{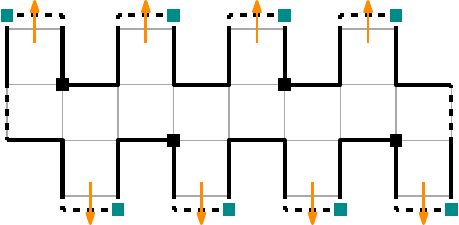}
	\caption{An adaption of the comb-like polyomino. The black vertices realizes an optimal guard set for the classic AGP, while the dark cyan set is optimal for the dispersive AGP.}
	\label{fig:introduction}
\end{figure}

\subsection{Our contributions}

In this paper, we introduce the dispersive art gallery problem and investigate it for vertex guards in polyominoes, i.e., orthogonal polygons whose vertices have integer coordinates. 
Our results are as follows.

\begin{itemize}
	\item We describe a (simple) thin polyomino where the minimum pairwise distance between any two guards in every feasible guard set is at most 3, see~\cref{lem:distance-three-necessary}.
	\item We give a worst-case optimal algorithm for placing a set of guards at the vertices of a simple polyomino such that the pairwise distances between any two guards are at least 3, see~\cref{thm:distance-three-sufficient}.
	\item It is \NP-complete to decide whether a pairwise distance of at least 5 can be guaranteed, see~\cref{thm:dispersion-distance-5-np-hard}.
	\item We describe a dynamic programming approach that computes a guard set that maximizes the minimum pairwise distance between any two guards for tree-shaped polyominoes, see~\cref{thm:optimal_for_thin_polyominos}.
\end{itemize}

\subsection{Previous work}

The famous question from Klee was answered relatively quickly by Chvátal~\cite{chvatal1975combinatorial}. 
Not least because of the beautiful proof from Fisk~\cite{fisk1978short} it is almost common knowledge that $\lfloor\nicefrac{n}{3}\rfloor$ guards are sufficient but sometimes necessary to monitor a simple polygonal region with~$n$~edges.  
Through their typical orthogonality, ``traditional'' galleries actually require less guards, i.e., for orthogonal polygons with $n$ vertices already $\lfloor\nicefrac{n}{4}\rfloor$ guards are sufficient, but also sometimes necessary~\cite{Hoffmann90,kahn1983traditional,o1983alternate}. 
However, finding the optimal solution even in simple polygons is proven to be \NP-hard by Lee and Lin~\cite{LeeL86}, and by Schuchardt and Hecker~\cite{SchuchardtH95} for simple orthogonal polygons.
In the special case of $r$-visibility, computing the minimum guard set is polynomial in orthogonal polygons~\cite{Biedl016,WormanK07}. 
More recently, Abrahamsen et al.~\cite{irrational-guards,agp-exist-r-complete} first showed that irrational guards are sometimes needed in an optimal guard set (in general and orthogonal polygons), and subsequently that the art gallery problem is actually $\exists \mathbb{R}$-complete.

Restricting the class of galleries to polyominoes intuitively makes the problem a lot easier. 
However, as shown by Biedl et al.~\cite{biikm-gp-11,biikm-agtp-12} the problem remains \NP-hard. 
On the positive side they showed that $\lfloor\nicefrac{m+1}{3}\rfloor$ point guards are always sufficient and sometimes necessary, where $m$ is the number of squares of the polyomino. 
Additionally, they give an algorithm for computing optimal guard sets in the case of thin polyomino trees.

By now, there are many variations of the classic art gallery problem. At least in two of them the number of placed guards is irrelevant, as it is also the case in our problem setting. These are the \textsc{Chromatic AGP}~\cite{erickson2010chromatic,EricksonL11,FeketeFHM014,IwamotoI20} where guards are associated by a color and no two guards of the same color class are allowed to have overlapping visibility regions, and the \textsc{Conflict-free chromatic AGP}~\cite{BartschiGMTW14,BartschiS14,hksvw-ccgoag-18} in which the overlapping constraint is relaxed in a way that at every point within the polygon a unique color must be visible. 
In both of these problems, only the number of used colors in a feasible guard set is of interest.

Other variations regard the region that has to be covered, e.g., the \textsc{Terrain Guarding Problem}~\cite{BonnetG19,KingK11}, or problems that arrive from restricting the visibility of the guards to cones of a certain angle, that can be summarized under the generic term of \textsc{Floodlight Problems}~\cite{AbelloESU98,BoseGLOSU97,CzyzowiczRU93,Estivill-CastroOUX95,ItoUY98a,NilssonOPSZ21,SteigerS98}.

Dispersion problems are related to packing problems and involve arranging a set of objects ``far away'' from one another, or choosing a subset of objects that are ``far apart''. 
These naturally arrive as obnoxious facility location problems (see, e.g., the surveys by Cappanera~\cite{cappanera1999survey}, or Erkut and Neuman~\cite{erkut1989analytical}), and as problems of distant representatives~\cite{FialaKP05}.
For more recent work in many different settings, e.g., in disks~\cite{DumitrescuJ12,FialaKP05}, or on intervals~\cite{BiedlLNRS21,LiW18}; see also~\cite{baur2001approximation,BenkertGKOW09,Cabello07,ChandraH01,FeketeM03,FormannW91,JiangBQZ04} for various other settings.

\subsection{Preliminaries}

We consider \emph{polyominoes}, that are orthogonal polygons formed by joining unit squares edge to edge. 
These unit squares are called \emph{cells}, and the edges of the cells are denoted as \emph{sides}.
The \emph{boundary}~$\partial \polygon$ of the polyomino $\polygon$ is the sequence of all cell sides each one lying between one cell from $\polygon$ and one cell not being part of~$\polygon$.
The \emph{vertices} of a polyomino $\polygon$ are the vertices of the boundary of $\polygon$. 
A~point~$p\in \polygon$ \emph{covers} or \emph{sees} another point $q \in \polygon$ if there is an axis-aligned rectangle defined by $p$ and $q$ that is a subset of $\polygon$. 
In the literature this notion of visibility is called \emph{$r$-visibility}. 
The area that is visible from a point $p$ is its \emph{visibility region}~$\visregion(p)$.
The~\emph{distance}~$d(p,q)$ between two points $p,q\in \polygon$ is given by the $L_1$ geodesic shortest path connecting these two points, i.e., the distance is measured entirely within the interior of $\polygon$. 
A~\emph{guard set}~$\guardset$ is a set of points of $\polygon$ such that every point of~$\polygon$ is covered by at least one point of $\guardset$. 
We will restrict ourselves to \emph{vertex guards}, i.e., guards that are placed on vertices of $\polygon$. 
The minimum over all pairwise distances between any two guards in a guard set~$\guardset$ is called its \emph{dispersion distance}.
The \emph{dual graph} of a polyomino $\polygon$ has a vertex for every cell of $\polygon$, and edges between vertices if their corresponding cells share a~side.
We~say that a polyomino is \emph{simple} if it has no holes, \emph{thin} if it does not contain a $2\times 2$ polyomino as a subpolyomino, and \emph{tree-shaped} if its dual graph is a tree.
We call a cell a \emph{niche} if it is a degree $1$ vertex in the dual graph of $\polygon$.

\section{Worst-case optimality}\label{sec:worstcase-optimal}

In this section we prove that a dispersion distance of $3$ is worst-case optimal for simple polyominoes. 
In particular, we construct thin polyominoes for which no guard set can have a larger dispersion distance than 3, and describe an algorithm that computes such guard sets for any simple polyomino.

\begin{lemma}\label{lem:distance-three-necessary}
	There are (simple) thin polyominoes such that every guard set has dispersion distance at most 3.
\end{lemma}

\begin{proof}
	Consider the dark magenta region in~\cref{fig:distance-three-necessary}. 
	Note that this region has to be guarded by a guard $g$ that is placed on one of the four vertices that are incident to this region.
	Let $\Pi$ be the one of the four niches that is closest to~$g$.
	The~guard $g'$ that covers $\Pi$ has distance at most $3$ to $g$.
\end{proof}

\begin{figure}[h]
	\centering
	\includegraphics[page=1, scale=.75]{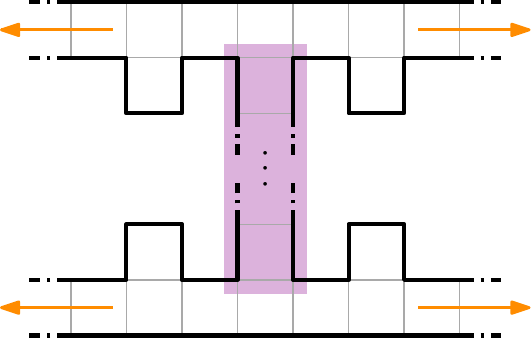}
	\caption{A simple, thin polyomino in that every guard set has dispersion distance at most 3.}
	\label{fig:distance-three-necessary}
\end{figure}

Note that the polyomino depicted in~\cref{fig:distance-three-necessary} can be used as a crucial ``building block'', i.e., it can be extended (as indicated by orange arrows) and therefore be used to construct arbitrarily large polyominoes 
in which the same upper bound holds.

In the remainder of this section we show that at least every \emph{simple} polyomino allows for a guard set with a dispersion distance of at least 3, implying worst-case optimality.

\begin{theorem}\label{thm:distance-three-sufficient}
	For every simple polyomino there exists a guard set that has dispersion distance at least 3.
\end{theorem}

We prove \cref{thm:distance-three-sufficient} constructively by giving an algorithm that constructs a guard set with dispersion distance of at least 3 in polynomial time. 
In a nutshell, the algorithm places guards greedily until the whole polyomino is guarded. 
The algorithm starts with a guard on an arbitrary vertex. 
Then the region that is visible from this guard is removed from the polyomino. 
This leads to a set of disjoint subpolyominoes that are guarded recursively, maintaining a distance of at least 3 between any two guards, see~\cref{fig:wco-preliminaries}.

\begin{figure}[h]
	\begin{subfigure}{0.5\columnwidth}
		\centering
		\includegraphics[page=1, scale=.80]{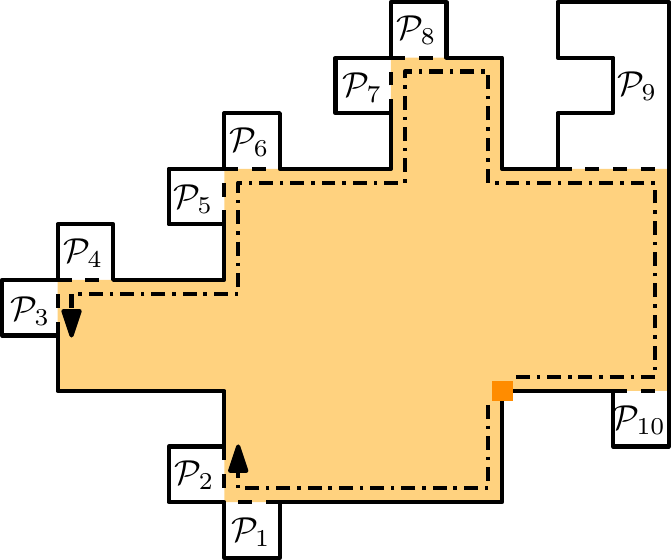}
		\caption{}
		\label{fig:wco-preliminaries-a}
	\end{subfigure}\hfill%
	\begin{subfigure}{0.5\columnwidth}
		\centering
		\includegraphics[page=2, scale=.80]{worst-case-optimality-new.pdf}
		\caption{}
		\label{fig:wco-preliminaries-b}
	\end{subfigure}
	\caption{(a) A guard with its visibility region (orange), and the subpolyominoes $\polygon_1, \dots,\polygon_{10}$. The corresponding gates $G_1$ and $G_2$ are clockwise, and $G_3,\dots,G_{10}$ are counterclockwise. (b) The recursion tree $T$ and a guard set (colored vertices) computed by our algorithm.} 
	\label{fig:wco-preliminaries}
\end{figure}

\subsection{Preliminaries for the algorithm}
Let $\polygon'$ be a subpolyomino of $\polygon$, i.e., $\polygon' \subseteq \polygon$. 
The boundary $\partial \polygon'$ of $\polygon'$ is the union of all sides being part of exactly one cell from $\polygon'$. 
Note that the definition of $\partial \polygon'$ does not depend on~$\polygon$. 
Assume that the guard $g$ cannot see the entire polygon $\polygon$, i.e., $\visregion(g) \neq \polygon$. 
By removing~$\visregion(g)$ from $\polygon$ we obtain $k \geq 1$ subpolyominoes $\polygon_1, \dots, \polygon_k \subset \polygon$, being maximal subsets of unit squares such that each subset forms an orthogonal polygon.
The~\emph{gate}~$G_i$ corresponding to~$\polygon_i$ is $\partial \visregion(g) \cap \partial \polygon_i$. 
Without loss of generality we assume $G_1,\dots,G_k$ to be ordered clockwise on~$\partial \polygon$ starting from~$g$.
The \emph{walls} of a gate~$G_i$ are the two sides from $\partial \polygon \setminus G_i$ being adjacent to $G_i$, see the red segments in~\cref{fig:recursive_calls}. 
Note that the first (second) wall of $G_i$ can lie on~$\partial \polygon_{i-1}$ ($\partial \polygon_{i+1}$) where from now on the indices $i+1$ and $i-1$ are considered modulo~$k$.

\subsection{Description of the algorithm}

Based on the preliminaries above we provide the details of our algorithm. 
As initialization, we consider a guard $g$ placed on an arbitrary vertex of the given polyomino $\polygon$.

\begin{description}
	\item[A recursion step]
	Consider the subshapes $\polygon_1, \dots, \polygon_k$ and the corresponding gates as defined above, see~\cref{fig:wco-preliminaries}(a).
	Let $\alpha$ and $\beta$ be the number of sides from $\partial \polygon$ when walking clockwise along $\partial \polygon$ from $g$ to $G_1$, and from $G_k$ to $g$, respectively. 
	Note that $\alpha,\beta \geq 1$.
	
	In the following we declare each gate to be \emph{(oriented) clockwise} or \emph{counterclockwise}. 
	For~this, we consider different cases regarding $k$, see~\cref{fig:case_distinction}.
	\begin{description}
		\item[(1)] If $k = 1$:
		\begin{description}
			\item[(1.1)] if $\alpha = 1$, we declare $G_1$ to be clockwise
			\item[(1.2)] otherwise, we declare $G_1$ to be counterclockwise.
		\end{description}
		\item[(2)] If $k = 2$:
		\begin{description}
			\item[(2.1)] if $\alpha = 1 = \beta$, we declare $G_1$ to be clockwise and $G_2$ to be counterclockwise,
			\item[(2.2)] if $\alpha = 1 < \beta$, we declare $G_1$ and $G_2$ to be clockwise,
			\item[(2.3)] if $\alpha > 1 = \beta$, we declare $G_1$ and $G_2$ to be counterclockwise.
		\end{description}
		\item[(3)] If $k \geq 3$, let $G_{\ell}$ be the first gate being not adjacent to its successor $G_{\ell+1}$, i.e., $G_{\ell}$ and $G_{\ell+1}$ are not sharing an endpoint. We declare $G_1,\dots,G_{\ell}$ to be clockwise and $G_{\ell+1},\dots,G_{k}$ to be counterclockwise, see~\cref{fig:wco-preliminaries-a}.
	\end{description}
	\begin{figure}[h]
		\centering
		\includegraphics[page=3, scale=0.8]{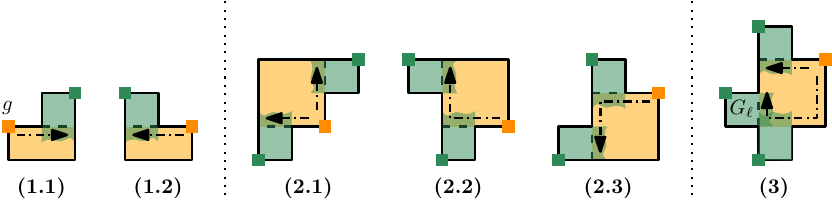}
		\caption{Case distinction for gate orientations (orange: guard $g$, green: guard $\overline{g}$ placed after $g$ and whose position is influenced by orientation of corresponding gate).}
		\label{fig:case_distinction}
	\end{figure}
\end{description} 

For each $\polygon_i$ we make a recursive call for covering $\polygon_i$ separately. In particular, we make the following distinction: A gate is \emph{parallel} (\emph{orthogonal}) when its walls lie parallel (orthogonal) to each other. 

\begin{figure}[h]
	\begin{subfigure}[b]{0.25\columnwidth}
		\centering
		\includegraphics[page=4, scale=.6]{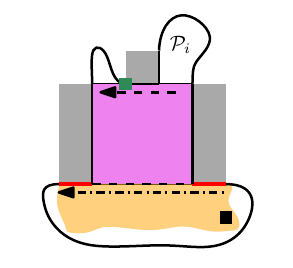}
		\caption{}
		\label{fig:recursive_calls-a}
	\end{subfigure}\hfill%
	\begin{subfigure}[b]{0.25\columnwidth}
		\centering
		\includegraphics[page=5, scale=.6]{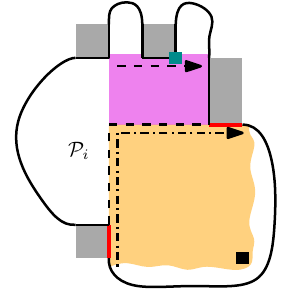}
		\caption{}
		\label{fig:recursive_calls-b}
	\end{subfigure}\hfill%
	\begin{subfigure}[b]{0.25\columnwidth}
		\centering
		\includegraphics[page=6, scale=.6]{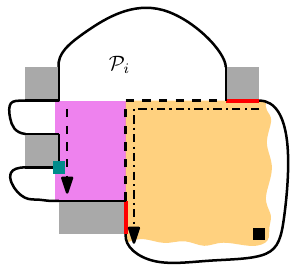}
		\caption{}
		\label{fig:recursive_calls-c}
	\end{subfigure}\hfill%
	\begin{subfigure}[b]{0.25\columnwidth}
		\centering
		\includegraphics[page=7, scale=.6]{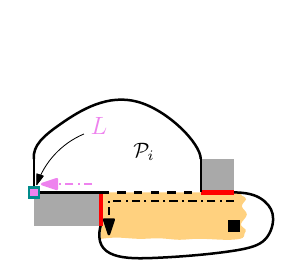}
		\caption{}
		\label{fig:recursive_calls-d}
	\end{subfigure}\hfill%
	\caption{Different placements of the guard $\overline{g}$ (green) depending on the gate's type and orientation (dashed line): (a)~A~parallel counterclockwise gate. (b) An orthogonal clockwise gate. (c) An orthogonal counterclockwise gate. (d)~An orthogonal counterclockwise gate, and $L$ degenerates to a single point.} 
	\label{fig:recursive_calls}
\end{figure}

\begin{description}
	\item[A recursive call for a parallel gate.] Without loss of generality, we assume that $G_i$ is horizontal and $\visregion(g)$ lies below $G_i$, see~\cref{fig:recursive_calls-a}.
	Let $T \subseteq \mathcal{P}_i$ be the axis aligned rectangle with maximal height and  bottom side $G_i$.
	If $G_i$ is clockwise (counterclockwise), we choose the guard $\overline{g}$ as an arbitrary vertex on the boundary of $T$ not lying on $G_i$ and not lying on the left (right) side of $T$. Note, that $\overline{g}$ always exists because $T$ has maximal height while ensuring $T$ to be contained inside the polyomino.
	Finally, we recurse on $\polygon:= \polygon_i \cup \visregion(\overline{g})$ and $g := \overline{g}$.

\item[A recursive call for an orthogonal gate.] 
	Without loss of generality, we assume that $\polygon_i$ lies above and to the left of~$G_i$, see~\cref{fig:recursive_calls-b,fig:recursive_calls-c}. 
	We distinguish two cases.
	\begin{description}
		\item[(1)] $G_i$ is counterclockwise: 
		Let $\ell \subseteq G_i$ be the vertical segment of $G_i$. 
		Note that if $G_i$ only consists of a horizontal segment, $\ell$ denotes the left endpoint of $G_i$, see~\cref{fig:recursive_calls-d}. 
		Let $L \subseteq \mathcal{P}_i$ be maximal rectangle with right side $\ell$, see~\cref{fig:recursive_calls-c,fig:recursive_calls-d}. 
		We choose $\overline{g}$ as a vertex from the boundary of $L$ not lying on $G_i$ 
		but from the left side of $L$.
		\item[(2)] $G_i$ is clockwise: 
		Consider by $t \subset G_i$ the horizontal segment of $G_i$. 
		Note that if $G_i$ only consists of a vertical segment, $t$ denotes the top endpoint of $G_i$.
		Let $T \subseteq \mathcal{P}_i$ be the maximal rectangle with bottom side $t$, see~\cref{fig:recursive_calls-b}. 
		We choose $\overline{g}$ as a vertex from the boundary of $T$ not lying on $G_i$ but from the top side of~$T$.
	\end{description}
	Note that $\overline{g}$ always exists because $T$ has maximal height in the first case and a maximal width in the second case while ensuring $T$ to be contained inside the polyomino. Finally, we again recurse on $\polygon:= \polygon_i \cup \visregion(\overline{g})$ and $g:= \overline{g}$.
	Intuitively speaking, considering the orientation used for a previously placed guard ensures its distance to be at least $3$ to the next placed guard; this is shown in~\cref{lem:dispersive_dist_three}, where the first and second case are special cases.
\end{description}

\subsection{Analysis of the algorithm}

We consider the recursion tree $T$ of our algorithm. 
In particular, each guard placed in the corresponding recursion step is a node in $T$. 
An edge between a father node $g$ and a child node $\overline{g}$ exists if $g$ creates a subpolyomino $\polygon_i$ causing a recursive call on $\polygon_i \cup \visregion(\overline{g})$ and $\overline{g}$. We~say that $g_2$ is a \emph{descendent} of $g_1$ if there is a sequence of nodes $g_1 = \overline{g}_1, \dots, \overline{g}_{\ell} = g_2$ such that $\overline{g}_i$ is the father of $\overline{g}_{i+1}$ for $i = 1,\dots, \ell-1$. 

For a clearer presentation, we say that $g_1$ is a \emph{child} of $g_2$, if $\ell = 1$.
As $\overline{g}$ is chosen from the segment resulting from pushing a vertical or horizontal line of $G_i$ until a vertex of $\polygon$ is hit for the first time, we obtain the following:

\begin{observation}\label{obs:covering_cells_adjacent_to_gate}
	All cells from $\polygon_i$ that share at least a point or a side with $G_i$ are seen by~$\overline{g}$. The corresponding cells are shown as dark green regions in~\cref{fig:placed_at_least_one_step_behind}.
\end{observation}

\begin{figure}[h]
	\begin{subfigure}[b]{0.25\columnwidth}
		\centering
		\includegraphics[page=18, scale=.6]{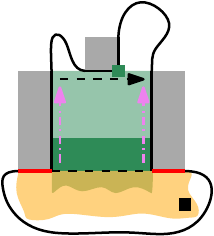}
		\caption{}
		\label{fig:placed_at_least_one_step_behind-a}
	\end{subfigure}\hfill%
	\begin{subfigure}[b]{0.25\columnwidth}
		\centering
		\includegraphics[page=19, scale=.6]{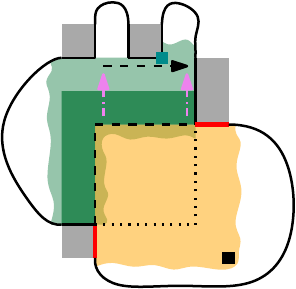}
		\caption{}
		\label{fig:placed_at_least_one_step_behind-b}
	\end{subfigure}\hfill%
	\begin{subfigure}[b]{0.25\columnwidth}
		\centering
		\includegraphics[page=20, scale=.6]{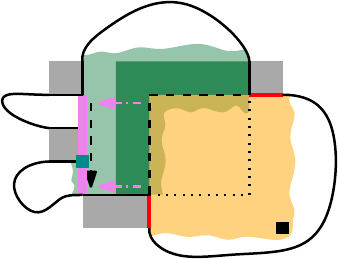}
		\caption{}
		\label{fig:placed_at_least_one_step_behind-c}
	\end{subfigure}\hfill%
	\begin{subfigure}[b]{0.25\columnwidth}
		\centering
		\includegraphics[page=21, scale=.6]{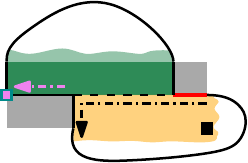}
		\caption{}
		\label{fig:placed_at_least_one_step_behind-d}
	\end{subfigure}\hfill%
	\caption{All cells from $\polygon_i$ that share at least a point or a side with $G_i$ are seen by $\overline{g}$.}
	\label{fig:placed_at_least_one_step_behind}
\end{figure}

As our algorithm recurses on $\polygon_i \cup \visregion(\overline{g})$, we obtain the following as a direct consequence of~\cref{obs:covering_cells_adjacent_to_gate}.

\begin{corollary}\label{lem:behind_gate}
	For each recursive call on $\polygon_i \cup \visregion(\overline{g})$ and $\overline{g}$ the guard $\overline{g}$ is placed inside $\polygon_i$ within a distance of at least $1$ to the corresponding gate $G_i$.
\end{corollary}

\begin{restatable}{lemma}{neighboredGates}\label{lem:neighbored_gates}
	Let $G_1$ and $G_2$ be two gates created in the same recursion step. If $G_1$ and $G_2$ share an endpoint, they have the same orientation.
\end{restatable}
\begin{proof}
	The proof follows the case distinction of the recursion step, and let $k$ be the number of gates created in the considered recursion step.	 
	As $k \geq 2$, Case (1) is not relevant. 
	If~$k=2$, the Cases (2.2) and (2.3) directly imply the same orientation of $G_1$ and $G_2$. 
	In Case (2.1) both gates are connected via one side of the boundary of $\polygon$, i.e., $\gamma_1 = \gamma_2$. 
	Thus, $G_1$ and $G_2$ cannot share an endpoint, see~\cref{fig:case_distinction}~(2.1). 
	Finally, let $k\geq 3$ gates be created during the considered recursion step. 
	If $G_1$ and $G_2$ share an endpoint, the description of the case ensures that they are oriented in the same direction. 
\end{proof}

We now consider the geometric form of gates and the positions of two gates relative to one another. 
Each gate contains either a single, or two segments. 
In the latter case, these lie adjacent and orthogonal. 
Note that while an orthogonal gate can consist of a single segment, see \cref{fig:recursive_calls-d}, a parallel gate cannot consist of two segments. 
For each segment $s$ of a gate, consider the maximal segment $S \subseteq s$ inside $\polygon$ dividing $\polygon$ into two polyominoes $\polygon_{\text{left}},\polygon_{\text{right}} \subset \polygon$ where $\polygon_{\text{left}}$ ($\polygon_{\text{right}}$) lies to the left (right) of $S$ in clockwise order. 
We say that the guard $g \in \polygon$ \emph{lies to the right} of $s$ because $g \in \polygon_{\text{right}}$, see~\cref{fig:parallel_gates_orthogonal-a}.

\begin{restatable}{lemma}{gatesRecursionShareEndpoint}\label{lem:adjacent_gates_parallel_and_orthogonal}
	If two gates $G_1, G_2$ created in the same recursion step share an endpoint, both $G_1$ and $G_2$ are parallel gates lying orthogonal to one another.
\end{restatable}
\begin{proof}
	The proof is by contradiction. 
	Assume that at least $G_1$ is orthogonal. 
	Two adjacent segments from different gates cannot be collinear, because otherwise there is a side from the boundary of $\polygon$ lying between cells from the polygon, see~\cref{fig:parallel_gates_orthogonal-b}. 
	As $G_1$ is orthogonal and adjacent to $G_2$, there is a sequence of consecutive segments $s_1,s_2,s_3$ from these gates, where $s_1,s_2$ and $s_2,s_3$ are adjacent to one another, see~\cref{fig:parallel_gates_orthogonal-c}. 
	As the guard $g$ lies to the right of $s_1$, $s_2$, and $s_3$, it sees at least one cell outside of $\visregion(g)$ being a contradiction.
\end{proof}

\begin{figure}[h]
	\begin{subfigure}[b]{0.165\columnwidth}
		\centering
		\includegraphics[page=12, scale=.75]{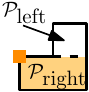}
		\caption{}
		\label{fig:parallel_gates_orthogonal-a}
	\end{subfigure}\hfill%
	\begin{subfigure}[b]{0.165\columnwidth}
		\centering
		\includegraphics[page=13, scale=.75]{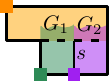}
		\caption{}
		\label{fig:parallel_gates_orthogonal-b}
	\end{subfigure}\hfill%
	\begin{subfigure}[b]{0.165\columnwidth}
		\centering
		\includegraphics[page=14, scale=.75]{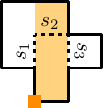}
		\caption{}
		\label{fig:parallel_gates_orthogonal-c}
	\end{subfigure}\hfill%
	\begin{subfigure}[b]{0.165\columnwidth}
		\centering
		\includegraphics[page=15, scale=.75]{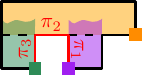}
		\caption{}
		\label{fig:parallel_gates_orthogonal-d}
	\end{subfigure}\hfill%
	\begin{subfigure}[b]{0.165\columnwidth}
		\centering
		\includegraphics[page=16, scale=.75]{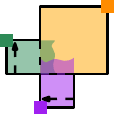}
		\caption{}
		\label{fig:parallel_gates_orthogonal-e}
\end{subfigure}\hfill%
	\begin{subfigure}[b]{0.165\columnwidth}
		\centering
		\includegraphics[page=17, scale=.75]{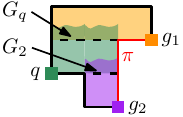}
		\caption{}
		\label{fig:parallel_gates_orthogonal-f}
\end{subfigure}\hfill%
	\caption{(a) A segment of a gate separates $\polygon$ into a left and a right polyomino. (b)~Two parallel gates lying collinear are not possible. (c)~Two gates laying adjacent where at least one is an orthogonal gate is not possible. (d)~A shortest path connecting two guards not being 
	children or descendents of each other where the corresponding gates are not adjacent. (e)~Two guards not being 
	children or descendents of each other where the corresponding gates are adjacent. (f)~Sequence of children.}
\label{fig:parallel_gates_orthogonal}
\end{figure}

We now analyze the dispersion distance of the guard set constructed by our approach based on \cref{lem:behind_gate,lem:neighbored_gates,lem:adjacent_gates_parallel_and_orthogonal}.

\begin{restatable}{lemma}{dispersiveDistanceThreeLowerBound}\label{lem:dispersive_dist_three}
	The constructed guard set has a dispersion distance of at least 3.
\end{restatable}
\begin{proof}
In order to prove the lemma we consider an arbitrary pair of placed guards~$g_1,g_2$~and distinguish three cases: 
(1) Neither $g_1$ is a descendent of $g_2$ nor vice versa. (2) $g_1$ is a descendent but not a child of $g_2$ or vice versa. (3)~$g_1$ is a child of $g_2$ or vice versa. In the following, we consider all three cases separately. The intuitions for the three cases are the following: (1) If $g_1$ and $g_2$ have the common father $g$ let $G_1,G_2$ be the corresponding gates. If $G_1,G_2$ are not adjacent these gates are within a distance of at least $1$. Hence, applying~\cref{lem:behind_gate} twice leads to a distance of at least~$3$, see~\cref{fig:parallel_gates_orthogonal-d}. If $G_1,G_2$ are adjacent,~\cref{lem:neighbored_gates} implies a distance of at least $3$, see~\cref{fig:parallel_gates_orthogonal-e}. If $g_1$ or~$g_2$ is not a child of $g$, similar arguments apply. (2) Applying~\cref{obs:covering_cells_adjacent_to_gate} yields two gates $G_1,G_2$ between $g_1$ and $g_2$ where each path between $G_1$ and $G_2$ has a length of $1$, see~\cref{fig:parallel_gates_orthogonal-f}. Finally, applying~\cref{lem:behind_gate} and the observation that $g_2$ does not lie on a gate caused by $g_2$ yields a distance of at least $3$. (3) Intuitively speaking~\cref{fig:placed_at_least_one_step_behind} implies that the same arguments as used in (2) apply to (3).

\begin{description}
\item[Neither $g_1$ is a descendent of $g_2$, nor vice versa.] 
First consider the case in which $g_1$ and $g_2$ are children of the same father $g$. 
Let $G_1$ and $G_2$ be the gates created by~$g$ corresponding to $g_1$ and $g_2$. 
Let $\pi$ be a shortest path connecting $g_1$ and $g_2$. 
Note that $\pi = (\pi_1,\pi_2, \pi_3)$ contains three subpaths, where $\pi_1$ connects $g_1$ and~$G_1$, $\pi_2$ connects $G_1$ and $G_2$, and $\pi_3$ connects $G_2$ and $g_2$. 
\cref{lem:behind_gate} implies that $\pi_1$ and $\pi_3$ have a length of at least $1$. 
If $G_1$ and $G_2$ do not share an endpoint, we obtain that $\pi_2$ has a length of at least~1, implying that $\pi$ has a length of at least~3, see~\cref{fig:parallel_gates_orthogonal-d}.

If $G_1$ and $G_2$ share an endpoint, \cref{lem:neighbored_gates} implies that $G_1$ and $G_2$ have the same orientation. 
Without loss of generality, assume that $G_1,G_2$ are oriented clockwise. 
Furthermore,~\cref{lem:adjacent_gates_parallel_and_orthogonal} implies that $G_1,G_2$ are parallel gates, whose segments lie orthogonal to one another. 
Without loss of generality, assume that $G_1,G_2$ are ordered clockwise and that $G_1$~($G_2$) is horizontal~(vertical), see~\cref{fig:parallel_gates_orthogonal-e}. 
As $G_1$ and $G_2$ are oriented clockwise, applying~\cref{lem:behind_gate} simultaneously to $g_1$ and $g_2$ implies that the $x$-coordinate of $g_1$ is at least one larger than the $x$-coordinate of $g_2$ and the $y$-coordinate of $g_1$ is at least two smaller than the $y$-coordinate of $g_2$. 
Thus, $g_1$ and $g_2$ have a distance of at least~$3$.

\item[$g_1$ is a descendent but not a child of $g_2$, or vice versa.] 
Without loss of generality, assume that $g_1$ is a descendent of $g_2$. 
Thus, there is at least one further guard $q$ being placed between $g_1$ and $g_2$, i.e., such that $q$ is a child or descendent of $g_1$ and $g_2$ is a child or descendent of $q$, see~\cref{fig:parallel_gates_orthogonal-f}. 
This implies that the shortest path $\pi$ connecting $g_1$ and $g_2$ has to cross the visibility region $\visregion(q)$ of $q$. 
\cref{obs:covering_cells_adjacent_to_gate} implies that this subpath of $\pi$ has a length of at least $1$. 
Let $G_q$ and $G_2$ be the gates between $g_1,q$ and $q,g_2$. 
\cref{lem:behind_gate} implies that the length of the subpath of $\pi$ connecting $G_2$ with $g_2$ is at least $1$. 
Furthermore, $g_1$ cannot lie on $G_q$ implying that the length of the subpath of $\pi$ connecting $g_1$ and $G_q$ also has a length of $1$. 
Hence, $\pi$ has a length of at least $3$.

\item[$g_1$ is a child of $g_2$, or vice versa.] 
Without loss of generality, assume that each segment of the gate $G_1$ corresponding to $g_1$ has a length of $1$, see~\cref{fig:wco-wlog-assumptions} for the different cases. 
In the case of a parallel gate, assume without loss of generality that $g_2$ lies adjacent to an endpoint $G_1$ resulting in a distance of $3$, see~\cref{fig:wco-wlog-assumptions-a}. 
In the case of an orthogonal gate, and that $G_1$ consist of two segments, assume without loss of generality that $g_2$ lies as close as possible to both segments of $G_1$ resulting in a distance of $4$ between $g_1$ and $g_2$, see~\cref{fig:wco-wlog-assumptions-b,fig:wco-wlog-assumptions-c}. 
Finally, in the case of an orthogonal gate that consist of a single segment, assume without loss of generality that $g_2$ lies on the wall collinear with $G_1$ resulting in a distance of at least $3$, see~\cref{fig:wco-wlog-assumptions-d}. 
\end{description}
This concludes the proof of~\cref{lem:dispersive_dist_three}.
\end{proof}

\begin{figure}[h]
	\begin{subfigure}[b]{0.25\columnwidth}
		\centering
		\includegraphics[page=8, scale=.75]{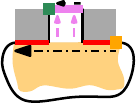}
		\caption{}
		\label{fig:wco-wlog-assumptions-a}
	\end{subfigure}\hfill%
	\begin{subfigure}[b]{0.25\columnwidth}
		\centering
		\includegraphics[page=9, scale=.75]{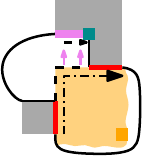}
		\caption{}
		\label{fig:wco-wlog-assumptions-b}
	\end{subfigure}\hfill%
	\begin{subfigure}[b]{0.25\columnwidth}
		\centering
		\includegraphics[page=10, scale=.75]{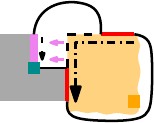}
		\caption{}
		\label{fig:wco-wlog-assumptions-c}
	\end{subfigure}\hfill%
	\begin{subfigure}[b]{0.25\columnwidth}
		\centering
		\includegraphics[page=11, scale=.75]{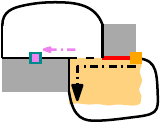}
		\caption{}
		\label{fig:wco-wlog-assumptions-d}
	\end{subfigure}\hfill%
	\caption{Assuming a position for $g$ decreasing at most its distance to $q$: (a) A parallel gate. (b) A clockwise oriented orthogonal gate made up of two segments. (c) A counterclockwise oriented orthogonal gate made up of two segments. (d)~An orthogonal gate made up of one segment.}
	\label{fig:wco-wlog-assumptions}
\end{figure}

As~\cref{lem:distance-three-necessary} provides an upper bound on the dispersion distance in simple polyominoes and~\cref{lem:dispersive_dist_three} the matching lower bound, these lemmas together prove~\cref{thm:distance-three-sufficient}.

\section{Computational complexity}\label{sec:hardness}

In this section we study the computational complexity of computing guard sets that maximize the smallest distance between all pairs of guards within the guard set.
In particular, we show the following.

\begin{restatable}{theorem}{hardnessTheorem}\label{thm:dispersion-distance-5-np-hard}
	In polyominoes it is \NP-complete to decide the existence of a guard set with a dispersion distance of 5.
\end{restatable}

First of all, note that the problem is obviously in \NP, as it is easy to verify whether a potential set of vertices of a given polyomino is in fact a guard set with a certain dispersion distance.

\begin{observation}\label{obs:problem-in-np}
	The dispersive art gallery problem for polyominoes with vertex guards is in \NP.
\end{observation}

In the remainder of this section we first give a high-level overview of the \NP-hardness reduction, followed by a description of the involved gadgets with analyses of their properties. We conclude the section by putting everything together and proving~\cref{thm:dispersion-distance-5-np-hard}.

\subsection{Outline of the NP-hardness reduction}

For proving \NP-hardness we make use of the problem \textsc{Planar Monotone 3Sat} that is shown to be \NP-complete by de~Berg and Khosravi~\cite{dbk-obspp-10}. 
This problem is a variant of the 3-satisfiability problem for which the literals in each clause are either all negated or all unnegated, and the corresponding variable-clause incidence graph is planar.

To this end, we will construct polyominoes that will represent variables and clauses. 
Because a variable may contribute to multiple clauses, we model a shape (see \emph{duplicator gadget}) that duplicates the given assignment. 
Furthermore, we describe simple shapes that are used to connect different subshapes, while maintaining the given assignment from the variables.
\cref{fig:hardness-overview} gives a high-level overview of the construction and the main gadgets.

\begin{figure}[h]
	\centering
	\includegraphics[page=12]{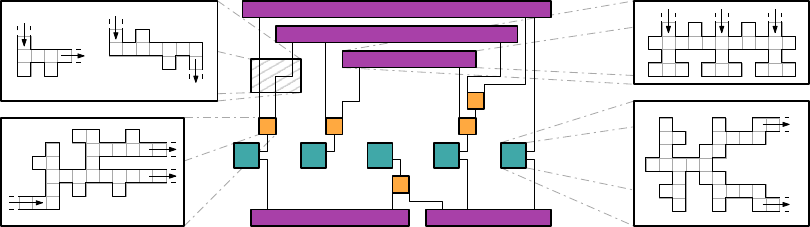}
	\caption{Symbolic overview of the \NP-hardness reduction. The depicted instance is due to the \textsc{Planar Monotone 3Sat} formula $\varphi=(x_1 \vee x_2 \vee x_4) \wedge (x_2 \vee x_4) \wedge (x_1 \vee x_4 \vee x_5) \wedge (\overline{x_1} \vee \overline{x_3}) \wedge (\overline{x_3} \vee \overline{x_4} \vee \overline{x_5})$. Variables are shown in dark cyan, clauses in magenta, duplicator gadgets in orange, and connectors as lines.}
	\label{fig:hardness-overview}
\end{figure}

The idea of the reduction is as follows: As shown in~\cref{fig:hardness-overview}, all gadgets have \emph{open~ends} (depicted by arrows) where they are connected to one another. 
These openings are distinguished as \emph{inputs} and \emph{outputs} depending on how the arrows are oriented.
Starting from the variable gadgets there is a sequence of outputs and inputs that ends in the clause gadgets. 
Through an intensive use of niches, we force the possible guard sets within each gadget to essentially two sets. 
In particular, the first set covers the output of a gadget from vertices of the gadget, and therefore it covers also the input of the next gadget in the sequence. 
The second set has to cover the gadget's input from its own vertices (because it is not already covered from before), and therefore cannot cover its output.
We use these constraints and observations to propagate variable assignments through the construction, i.e., to obtain a \emph{guarding direction}.

\subsection{Setting up the gadgets}

We now give the description of the involved gadgets and we prove several lemma that will then put together to yield a proof of~\cref{thm:dispersion-distance-5-np-hard}.

\subsubsection*{Variable gadget}

The variable gadget is depicted in~\cref{fig:variablegadget}. As we have to represent the values true and false, the gadget is constructed in such a way that it allows for two guard sets with a dispersion distance of 5. 
In addition, we make sure that no guard set with a larger dispersion distance exists.
In order to do so, we force guards to unique positions within the magenta regions, which then results in the fact that the positions of all other guards are restricted and partitioned into two disjoint sets.
	
	\begin{figure}[h]
		\begin{subfigure}{0.5\columnwidth}
			\centering
			\includegraphics[page=1, scale=.75]{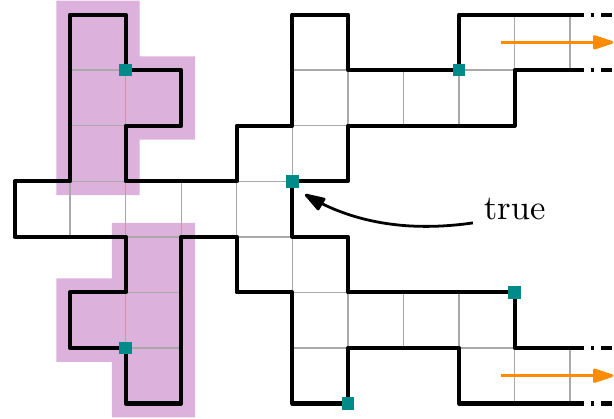}
			\caption{}
			\label{fig:variablegadget-true}
		\end{subfigure}\hfill%
		\begin{subfigure}{0.5\columnwidth}
			\centering
			\includegraphics[page=2, scale=.75]{hardness.pdf}
			\caption{}
			\label{fig:variablegadget-false}
		\end{subfigure}
		\caption{The figure shows the variable gadget, and highlighted regions that are used in~\cref{claim:variable-gadget-at-most-5,claim:variable-gadget-exactly-two-sets}. (a)~shows the unique guard set representing a truth assignment, while (b) shows the respective set for a false assignment.}
		\label{fig:variablegadget}
	\end{figure}

\begin{lemma}\label{claim:variable-gadget-at-most-5}
	Within the variable gadget, no guard set has a dispersion distance larger than~5.
\end{lemma}

\begin{proof}
	Consider the dark magenta regions (also called ``T-shapes'') in~\cref{fig:variablegadget-true}.
	No guard set with at least two guards that are placed exclusively on vertices of the T-shape realizes a dispersion distance larger than 5. 
	The only vertex that could partly cover such a T-shape from the outside is itself a vertex from another T-shape.
	Therefore, each region has to be covered uniquely from within it.
	Thus, the largest possible distance between these guards is 5, as shown by cyan squares.
\end{proof}

\begin{lemma}\label{claim:variable-gadget-exactly-two-sets}
	Within the variable gadget there are exactly two guard sets realizing a dispersion distance of 5.
\end{lemma}

\begin{proof}
	As the guard placement within the T-shapes is unique (see~\cref{claim:variable-gadget-at-most-5}), the only variability lies within the dark magenta region shown in~\cref{fig:variablegadget-false}. 
	However, by maintaining a distance of 5 to the necessary guards placed within the T-shapes, there are exactly two vertices that remain for covering this region.
	By choosing one, all the other positions follow uniquely.
	Hence, there are exactly two guard sets that realize a dispersion distance of 5.
\end{proof}

\subsubsection*{Clause gadget} 

The clause gadget is depicted in~\cref{fig:clausegadget}. 
The overall idea of this gadget is that it does not allow for a guard set with a dispersion distance of at least 5 if guards have to be placed only on vertices of this subshape. 
Hence, some specific cells have to be already covered from outside the shape, what will be related to satisfying the clause. 

	\begin{figure}[h]
	\begin{subfigure}{0.5\columnwidth}
		\centering
		\includegraphics[page=5, scale=.75]{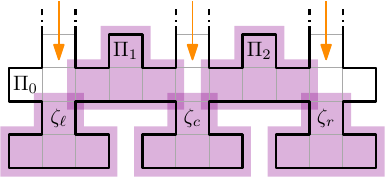}
		\caption{}
		\label{fig:clausegadget-three}
	\end{subfigure}\hfill%
	\begin{subfigure}{0.5\columnwidth}
		\centering
		\includegraphics[page=6, scale=.75]{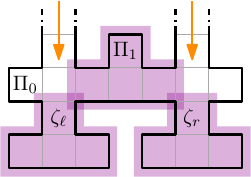}
		\caption{}
		\label{fig:clausegadget-two}
	\end{subfigure}
	\caption{(a) The depicted shape represents a clause gadget containing three literals, while (b)~shows a clause with two literals. The cells labeled with $\zeta$ can be covered from outside the gadget.} 
	\label{fig:clausegadget}
	\end{figure}

Note that a clause gadget ``contains'' basically two types of T-shapes, as shown in~\cref{fig:clausegadget}.
We call them \emph{prospects} if they can partly be covered from outside (i.e., these where a cell is labeled with $\zeta$), and \emph{checkers} otherwise.

\begin{lemma}
	There is no guard set with a dispersion distance of at least 5 for the shape representing the clause gadget when placing guards only within this shape.
\end{lemma}

\begin{proof}
	We will only argue this in detail for the case that the clause contains three literals; similar arguments hold for the remaining case.
	Consider one of the colored T-shapes in~\cref{fig:clausegadget-three}. 
	Only a single guard can be placed within such a region if a dispersion distance of at least 5 is required. 
	Consider three consecutive T-shapes, such that two of them are prospects.
	The shortest path connecting the six potential guard locations has a length of~9.
	Hence, no guard set with a dispersion distance of~5 exists.
\end{proof}

\begin{lemma}
	If at least one cell of the clause gadget is covered from outside the gadget, a~feasible guard set with a dispersion distance of 5 exists.
\end{lemma}

\begin{proof}
	Again, we will only argue the more complicated case, i.e., a clause containing three literals.
	For this, consider the marked region in~\cref{fig:clausegadget-three} and distinguish the following.
	
	First assume that the central connector is covered from outside the gadget. 
	Thus, in particular cell $\zeta_c$ is already covered, and we can place a guard in a bottom corner of this prospect.
	This results in two disjoint pairs of uncovered T-shapes, such that each of these pairs share a vertex of the shape.
	Without loss of generality, consider the left pair.
	Two guards can be placed at the bottom left vertex of $\zeta_\ell$ and bottom right of $\Pi_1$ with a  distance of 5.
	Because the guard that covers $\Pi_1$ already covers the bottom part of the other checker, we can place a guard in its niche, i.e., at the top right vertex of $\Pi_2$.
	Placing a guard at the bottom right vertex of $\zeta_r$ completes the guard set.
	
	On the other hand, without loss of generality, let the left connector be already covered, i.e., the cell $\zeta_\ell$ is covered. 
	A guard can be placed in the leftmost niche $\Pi_0$ to cover the bottom part of both checker. 
	Therefore, we only have to cover $\Pi_1$ and $\Pi_2$ by placing guards at their respective top vertices.
	It~remains to cover the prospects.
	Because guards are placed at the top vertices of the checker's niches, we can place guards in appropriate distances to obtain a guard set with dispersion distance~5.
\end{proof}

Due to the respective embedding of the overall shape it may be necessary to enlarge the clause gadget, see~\cref{fig:clausegadget-wide}. 

	\begin{figure}[h]
	\begin{subfigure}{0.5\columnwidth}
		\centering
		\includegraphics[page=7, scale=.75]{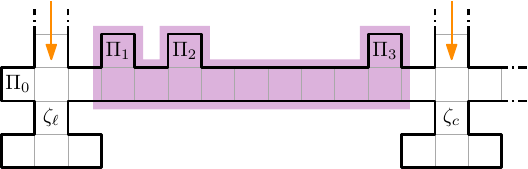}
		\caption{}
		\label{fig:clausegadget-wide-three}
	\end{subfigure}\hfill%
	\begin{subfigure}{0.5\columnwidth}
		\centering
		\includegraphics[page=8, scale=.75]{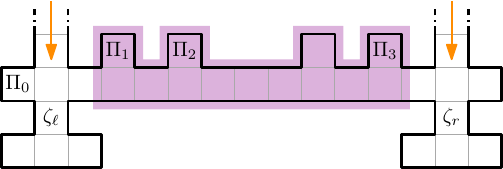}
		\caption{}
		\label{fig:clausegadget-wide-two}
	\end{subfigure}
	\caption{(a) shows the left part of an enlarged clause gadget containing three literals, while (b)~shows the respective enlarged shape for clauses containing two literals.} 
	\label{fig:clausegadget-wide}
\end{figure}
\clearpage
\begin{lemma}\label{lem:enlarged-gadget}
	A clause gadget can be enlarged in a way that all functionalities are maintained.
\end{lemma}

\begin{proof}
	If the clause contains three literals, we replace the T-shape checker by the colored region in~\cref{fig:clausegadget-wide-three}. 
	Note that this region is mirrored vertically along the center connector, and that the region between $\Pi_2$ and $\Pi_3$ can be enlarged arbitrarily.
	
	A crucial observation is that the niches $\Pi_1$ and $\Pi_3$ coincide in the short clause gadget, and therefore apply the same restrictions as before. 
	The additional niche $\Pi_2$ guarantees that we cannot place a guard at the bottom right vertex of~$\Pi_1$, assuming that the clause is not satisfied through an assignment.
	
	If the clause only contains two literals, the T-shape checker will be replaced by the colored region in~\cref{fig:clausegadget-wide-two}. 
	The correctness follows analogously.
\end{proof}
	
\subsubsection*{Duplicator gadget} 

Because a variable may contribute to more than one clause, we need to duplicate the respective assignment. 
For this purpose, we construct the duplicator gadget that is depicted in~\cref{fig:duplicatorgadget}. 
It works as follows: if the incoming connector is covered from outside the gadget, both outgoing connectors can be covered from within the gadget. 
Similarly, if the incoming connector has to be covered from within the gadget, the outgoing connectors must be covered from outside the gadget.

	\begin{figure}[h]
		\begin{subfigure}{0.5\columnwidth}
			\centering
			\includegraphics[page=4, scale=.75]{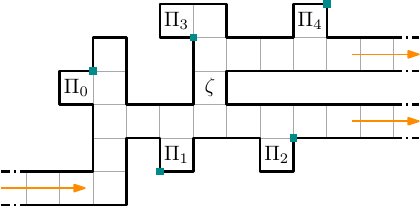}
			\caption{}
			\label{fig:duplicatorgadget-true}
		\end{subfigure}\hfill%
		\begin{subfigure}{0.5\columnwidth}
			\centering
			\includegraphics[page=3, scale=.75]{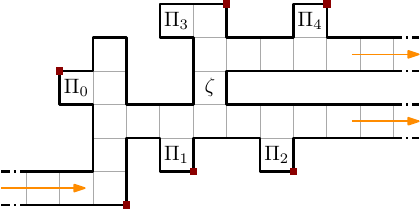}
			\caption{}
			\label{fig:duplicatorgadget-false}
		\end{subfigure}
		\caption{The figure shows the duplicator gadget. (a) shows a set of guards duplicating a true assignment, while (b) shows the respective guard set for a false assignment.}
		\label{fig:duplicatorgadget}
	\end{figure}

\begin{lemma}\label{lem:duplicator-is-correct}
	The duplicator gadget is correct, i.e., any output is equal to the input.
\end{lemma}

\begin{proof}
	First consider the situation given in~\cref{fig:duplicatorgadget-true}. 
	Because the incoming connector is covered from the outside, we want to cover the outgoing connectors from the inside. 
	We will argue that the configuration in the marked region is unique and fulfills the requirements. 
	Because of niche $\Pi_3$, covering $\Pi_1$ by a guard placed at vertices of $\zeta$ is not possible, and because of $\Pi_0$ and $\Pi_2$ the guard covering $\Pi_1$ is uniquely defined. 
	Because this position is fixed, all other positions follow.
	
	Now consider the situation in~\cref{fig:duplicatorgadget-false}. 
	The incoming connector has to be covered from the inside.
	Because of $\Pi_0$, the position of the guard covering the incoming connector is uniquely defined. 
	Therefore, there are two positions left to cover the niche $\Pi_1$; however, because of $\Pi_3$, we cannot choose the vertex of $\zeta$. 
	It follows that the positions for guarding $\Pi_1$ and $\Pi_2$ are uniquely defined. 
	Because $\zeta$ cannot be covered from the outside, the guard covering this square is also uniquely defined and also covers $\Pi_3$ simultaneously. 
	This again leaves only a single position to cover $\Pi_4$. 
	Overall, this leaves some squares of the outgoing connectors uncovered, so that they have to be covered from outside the gadget.
\end{proof}

Because a necessary condition to the problem is that the guard set have to cover the polyomino completely, we have to ensure that visibility regions that are induced by guards from clause gadgets do not interfere the assignment given due to guards from within the variable gadgets. 
So, if a clause $C_i = (x_j, x_k, x_\ell)$ is satisfied by say $x_j$, we have to make sure that other clauses containing $x_k$ or $x_\ell$ but not $x_j$ are not become automatically satisfied by \emph{backward guarding} from guards in $C_i$.
Preventing this is also the job of the duplicator gadget.

\begin{lemma}\label{lem:duplicator-cannot-flip}
	Backward guarding of an output of the duplicator gadget cannot result in covering the other output from within the gadget.
\end{lemma}

\begin{proof}
	Consider without loss of generality that the duplicator gadget propagates false from the respective variable. 
	A critical situation would occur if the assignment could be flipped within a duplicator gadget due to the coverage coming from a clause gadget, i.e., propagating true to the other output. 
	
	As argued above, due to the positions of the niches, the positions for guards are highly restricted.
	The coverage from the outside does not cover any of these niches.
	Therefore, it does not change the possible set of guard positions.
\end{proof}

\subsubsection*{Connector gadget} 

Now that we have the main components, we need to connect them. 
For this, we introduce two different connector gadgets, see~\cref{fig:connector}.
	\begin{figure}[h]
		\begin{subfigure}[b]{0.5\columnwidth}
			\centering
			\includegraphics[page=9, scale=.75]{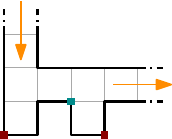}
			\caption{}
			\label{fig:connector-L}
		\end{subfigure}\hfill%
		\begin{subfigure}[b]{0.5\columnwidth}
			\centering
			\includegraphics[page=11, scale=.75]{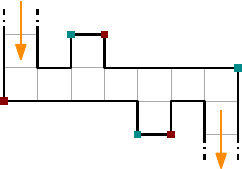}
			\caption{}
			\label{fig:connector-Z}
		\end{subfigure}
		\caption{(a) shows an L-connector, and (b) shows a Z-connector. The dark cyan and red colored guard sets propagate whether a variable is set to true or false, respectively.} 
		\label{fig:connector}
	\end{figure}

\begin{lemma}\label{lem:connector-gadgets}
	All connector gadgets fulfill the property that either the input, or the output can be guarded from within the gadgets by a guard set with a dispersion distance of~5.
\end{lemma}

\begin{proof}
	As these gadgets connect the previous ones, we distinguish between the cases that their input is already covered or not. Remark that if the input is already covered, we want to cover the output within the connector, and vice versa.
	
	We prove this by providing specific sets of guards, regarding the different settings.
	For the case that the input is already covered, consider the dark cyan placed guards. 
	The distance between guards is at least 5, and everything is covered.
	For the case that the input has to be covered, consider the red guarding positions.
	The placement of the niches force the position of the guard that covers the input. 
	All other positions follow uniquely. 
	It is easy to see that no more guards can be placed, so the output remains uncovered.
\end{proof}

\subsection{Completing the proof}

In the previous subsection we described several gadgets that will now be used to construct polyominoes $\polygon_\varphi$ as instances of the \textsc{Dispersive Art Gallery Problem} from a Boolean formula $\varphi$ that is an instance of \textsc{Planar Monotone 3SAT}. 
This yields a proof for~\cref{thm:dispersion-distance-5-np-hard}, which we restate here.

\hardnessTheorem*

\begin{proof}
	As already mentioned in~\cref{obs:problem-in-np}, the problem is obviously in~\NP.
	
	To show that the problem is \NP-hard, we reduce from \textsc{Planar Monotone 3Sat}. 
	For a given formula $\varphi$ we construct an instance $\polygon_\varphi$ of \textsc{dispersive AGP} as follows; again, see~\cref{fig:hardness-overview} for the high-level idea of the construction.
	Consider the rectilinear embedding of the graph given by $\varphi$. 
	For every variable, we place a variable gadget horizontally in a row.
	Each clause is represented by a clause gadget. 
	Due~to the rectilinear embedding, we can place them vertically behind one another, and expand them appropriately if necessary as shown above.
	Without loss of generality, we place the clauses containing only unnegated literals above the variables, and below otherwise.
	If a~literal~$x_i$ occurs in $m_i$ many clauses, we construct $m_i-1$ duplicator gadgets between vertically between clauses and variables.
	We properly place a set of connector gadgets to connect variables to duplicator gadgets, as well as the outputs of duplicator gadgets to respective inputs, and duplicator gadgets to the respective clauses. Note that variables are connected to clauses if they contribute only to a single clause.
	
	\begin{description}
		\item[If $\varphi$ is satisfiable, then there is a guard set with dispersion distance 5 for $\polygon_\varphi$.]
	\end{description}
	
	Consider a satisfying assignment of $\varphi$. 
	A guard set with a dispersion distance of 5 for $\polygon_\varphi$ can be constructed as follows: 
	From the given assignment of the variable $x_i$ the respective set of guards within the variable gadget is chosen. 
	For every connector and duplicator gadget, there is a set of guards that maintains the assignment. 
	Because we propagate the satisfying assignment through the gadgets, at least one literal satisfies each clause. 
	Hence, we can choose guards within each clause gadget that has dispersion distance of 5, because in each of these gadgets at least one of the cells are covered from the outside.
	
	\begin{description}
		\item[If there is a guard set with dispersion distance 5 for $\polygon_\varphi$, then $\varphi$ is satisfiable.]
	\end{description}
	
	Consider a guard set for $\polygon_\varphi$ that has a dispersion distance of~5. 
	As argued above, at least one cell of each clause gadget are covered from outside of the respective gadget, because otherwise there is no such desired guard set.
	Furthermore, there is no guard set for the variable gadget that has a dispersion distance larger than 5, and there are only two sets that realize this pairwise minimum distance. 
	For every path from variables to clauses, the duplicator and connector gadgets provide specific locations for guards that maintain a dispersion distance of~5.
	Hence, the guards within the variable gadget of $\polygon_\varphi$ realize a satisfying assignment for~$\varphi$.
	
	This concludes the proof.
\end{proof}

\begin{figure}[p]
	\centering
	\includegraphics[page=1, scale=0.65]{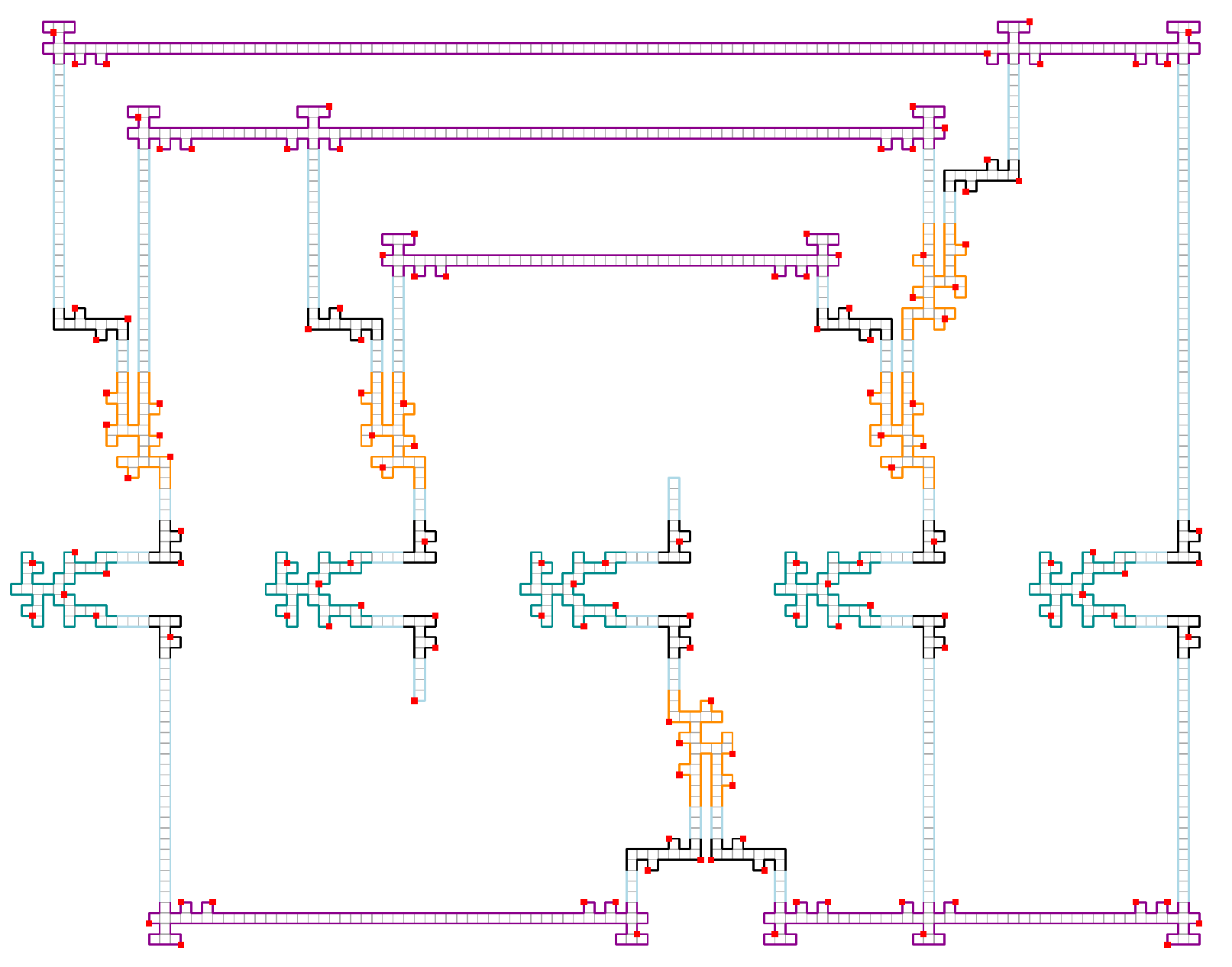}
	\caption{The polyomino $\polygon_\varphi$ as an instance for the dispersive art gallery problem derived from the Boolean formula $\varphi=(x_1 \vee x_2 \vee x_4) \wedge (x_2 \vee x_4) \wedge (x_1 \vee x_4 \vee x_5) \wedge (\overline{x_1} \vee \overline{x_3}) \wedge (\overline{x_3} \vee \overline{x_4} \vee \overline{x_5})$. Different gadgets are colored according to~\cref{fig:hardness-overview}; light blue boundary denote simple horizontal and vertical connecting paths. The~red guard set has a dispersion distance of 5 due to the assignment $(x_1,x_2,x_3,x_4,x_5) \rightarrow (0,1,1,1,0)$.}
	\label{fig:reduction-true}
\end{figure}
\begin{figure}[p]
	\centering
	\includegraphics[page=2, scale=0.65]{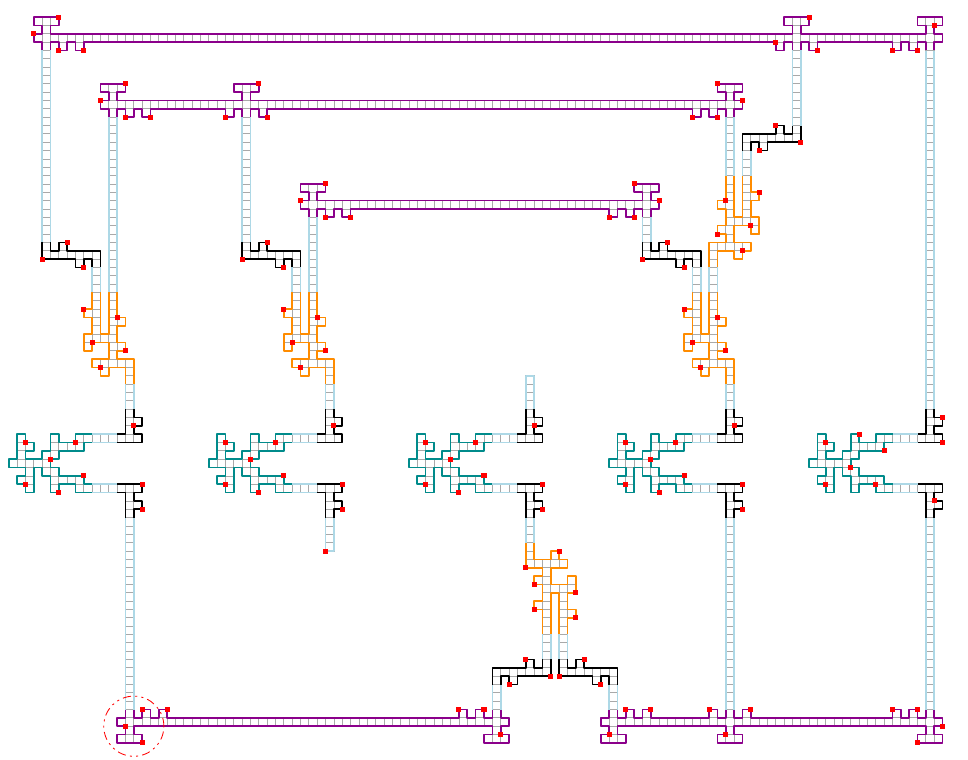}
	\caption{The figure depicts the same instance $\polygon_\varphi$ as in~\cref{fig:reduction-true}. The~red guard set has a dispersion distance of 4 due to the assignment $(x_1,x_2,x_3,x_4,x_5) \rightarrow (1,1,1,1,0)$; in particular the fourth clause is not satisfied, see red circle in the lower left.}
	\label{fig:reduction-false}
\end{figure}

\section{Optimality for tree-shaped polyominoes}\label{sec:tree-shapes}

While computing guard sets with maximum dispersion distance is \NP-hard in general, we present a linear-time algorithm to compute optimal solutions in tree-shaped polyominoes.
Recall that a polyomino $\polygon$ is tree-shaped if the dual graph of $\polygon$ is a tree. In particular, these polyominoes do not contain a $2\times 2$ subpolyomino.
\clearpage
\begin{theorem}\label{thm:optimal_for_thin_polyominos}
	Given a tree-shaped polyomino $\polygon$ with $n$ vertices, there is an $O(n)$ dynamic programming approach for computing guard sets of maximum dispersion distance.
\end{theorem}

\begin{figure}[htb]
	\begin{subfigure}{0.33\columnwidth}
		\centering
		\includegraphics[page=1, scale=.6]{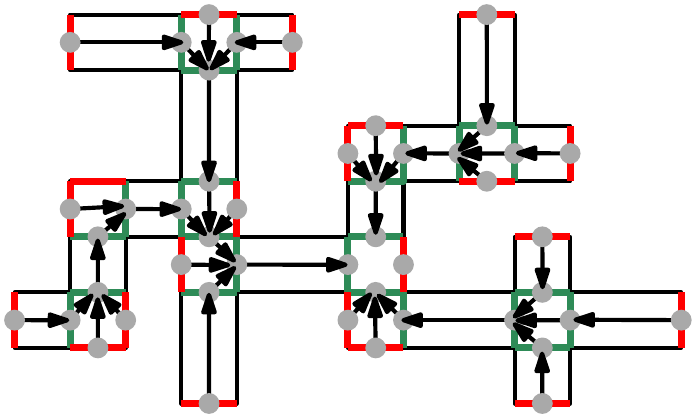}
		\caption{}
		\label{fig:dp-a}
	\end{subfigure}\hfill%
	\begin{subfigure}{0.33\columnwidth}
		\centering
		\includegraphics[page=3, scale=.6]{dp.pdf}
		\caption{}
		\label{fig:dp-c}
	\end{subfigure}\hfill%
	\begin{subfigure}{0.33\columnwidth}
		\centering
		\includegraphics[page=2, scale=.6]{dp.pdf}
		\caption{}
		\label{fig:dp-b}
	\end{subfigure}
	\caption{(a) The directed tree (black edges) and borders (red and green edges) of a tree-shaped polyomino. (b) The unique path from a border $b$ to a side of the root cell. (c) An optimal guard set (black squares) generated by our approach, with seeing directions (orange arrows).}  
	\label{fig:dp}
\end{figure}

We start by providing the main structure used in our algorithm, i.e., \emph{borders}; these are defined as follows: 
Let~$R \subseteq \polygon$ be the set of all maximal rectangles, see~\cref{fig:dp}. 
Note that, $R$ covers $\polygon$.
A side $s$ of a rectangle is an \emph{inner border} if $s \notin \partial \polygon$, see the green segments in~\cref{fig:dp-a}. 
The two sides of each rectangle having length $1$ are called \emph{outer borders}, see the red segments in~\cref{fig:dp-a}. 
Note that every outer border is on the boundary of $\polygon$, and every border is either an inner or an outer border.

Our dynamic programming approach follows a tree structure $T$ induced by the above-mentioned borders as follows: 
Let $c \in \polygon$ be an arbitrary cell containing at least two borders lying orthogonal to each other. 
In the following the cell~$c$ will contain the border that is going to be the root of $T$. 
Such a cell $c$ always exists as long as $\polygon$ is not a single $1 \times m$ rectangle for $m \in \mathbb{N}$. 
Furthermore, let $R'$ be the set of rectangles induced by the arrangement of $R$, i.e., $R'$ partitions $\polygon$. 
Thus, $c \in R'$. Let $r' \in R'$. In the following we exclusively use the term \emph{side of $r'$} for a side of $r'$ having a length of~$1$. 
Note that there are cells being rectangles $r' \in R'$. 
Thus, each rectangle $r' \in R'$ has even two or four sides. 
We define the tree $T = (V,E)$ where $V$ is the set of all borders as follows: 
As $\polygon$ is thin, for each outer border $b$ there is a unique sequence of rectangles $r_1, \dots,r_k \in R'$ such that (1) $b$ is a side of~$r_1$, (2) $r_i,r_{i+1}$ share a side $b_i$ that is an inner border for $i = 1,\dots,k-1$, and (3) $r_k$ is a side of $c$, see~\cref{fig:dp-c}. 
We define $E$ fulfilling $(b,b_1) \in E$ and $(b_i,b_{i+1}) \in E$ for $i \leq k-1$ for each border $b$.
A border~$b$ connects two positions $p_1,p_2 \in \polygon$ of which at least one is a vertex of $\polygon$, i.e., a possible position for a guard. 
Starting from a leftmost vertex of $\polygon$ with minimal $y$-coordinate, we consider all positions being part of a border to be ordered clockwise on $\partial \polygon$. 
We say that $p_1$ is \emph{smaller than} $p_2$ if $p_1$ is an (indirect) predecessor of $p_2$ in this order.

A key observation for our approach is the following:

\begin{observation}\label{obs:distance_difference-app}
	Let $p_1,p_2$ be the positions of a border and let $p_1$ be smaller than $p_2$. 
	If~$d_1,d_2$ denote the respective shortest distances to a guard of a given guard set, it holds that~$d_1-d_2 \in \{ -1,0,1 \}$.
\end{observation}

In the context of~\cref{obs:distance_difference-app}, we define the \emph{order} of $p_1,p_2$ as $d_1-d_2$.
Another crucial observation addresses in which way borders are seen by guards of a given guard set $\guardset$. 
We~say that an inner border $b$ is \emph{seen} by a guard $g$ when $b \subset \visregion(g)$.

\begin{observation}\label{obs:seen_border}
Let $\guardset$ be a guard set for $\polygon$. Then, for every inner border $b$ of $\polygon$ there is a guard $g\in \guardset$ that sees $b$.
\end{observation}

Let $b$ be an arbitrary inner border, and let $\polygon_1,\polygon_2 \subset \polygon$ be the two maximal nonoverlapping polyominoes sharing $b$ such that $c \subseteq \polygon_2$. 
We say that guard $g$ \emph{lies below (above) $b$} when $g \in \polygon_1$ ($g \in \polygon_2$). 
Motivated by this, we say that an inner border $b$ is \emph{seen from below (above)} when $b$ is seen by a guard $g$, and $g$ lies below (above)~$b$.

A \emph{state} of a border $b$ containing the positions $p_1,p_2$ is a triple $(O,M, D)$ where $O$ denotes the order of~$b$, $M$ is the subset of $\{p_1,p_2\}$ indicating which positions are chosen as guards, and $D \in \{ \text{below}, \text{above} \}$ the \emph{seeing direction} of~$b$.
The \emph{score} $s$ of $b$ is defined as the shortest distance of $p_1$ and $p_2$ to a guard lying below $b$.

Each border $b$ is associated with a set $S_b$ of \emph{correct pairs} being made up of a state and a respective score. 
For an outer border $b$ the initialization of $S_b$ depends on which positions $p_1,p_2$ of $b$ are vertices of $\polygon$, i.e., allowed to be chosen as guards. 
Let $p_1$ be smaller than $p_2$. We define $S_b$ for an outer border as follows:

\begin{itemize}
	\item If only $p_1$ is a vertex:
	\begin{eqnarray*}
		\{ & ((-1, \{ p_1 \}, \text{below}), 0),&\\
		   & ((\phantom{-}0, \{\phantom{p_1} \}, \text{above}), \infty) & \}
	\end{eqnarray*}
	\smallskip
	\item If only $p_2$ is a vertex:
	\begin{eqnarray*}
		\{ & ((1, \{ p_2 \}, \text{below}), 0),&\\
		   & ((0, \{\phantom{p_1} \}, \text{above}), \infty) & \}
	\end{eqnarray*}
	\smallskip
	\item If $p_1$ and $p_2$ both are vertices:
	\begin{eqnarray*}
		\{ & ((-1, \{ p_1 \}\;\phantom{p_2}, \text{below}), 0),&\\
		   & ((\phantom{-}1, \phantom{p_1,}\;\{p_2 \}, \text{below}), 0) ,&\\
		   & ((\phantom{-}0, \{p_1, p_2 \}, \text{below}), 0), &\\
		   & ((\phantom{-}0, \{\phantom{p_1,p_2} \}, \text{above}), \infty) & \}
	\end{eqnarray*}
\end{itemize}

So, e.g., $((-1,\{ p_1\}, \text{below}),0)$ represents that a guard is placed in position $p_1$ but not in~$p_2$, $((0,\{ p_1,p_2\}, \text{below}),0)$ that a guard is placed in both $p_1$ and $p_2$, and $((0,\{ \}, \text{above}),\infty)$ denotes that a guard is placed neither in $p_1$ nor $p_2$. 

For every inner border $b$ the set $S_b$ is initialized as the empty set.
The set of possible states for an inner border is the set of all combinations of values regarding $O$, $M$, and $D$. 
Based on the initialization of the outer borders, we compute the sets of all correct states for the inner borders in the order induced by the directed tree $T$. 
In particular, we initially mark all leaves, i.e., outer borders as \emph{processed} and the remaining vertices of $T$ as \emph{unprocessed}.
Let $w$ be a vertex of~$T$ whose children $v_1,\dots,v_k$, for $k \in \{1,2,3 \}$, are all processed. 
We compute the set of all correct states of $w$ and the respective optimal score values individually depending on the number of children of $w$. 
If $w$ has one child for each correct pair $((O_v,M_v,D_v),s_v)$ of $v$, we add all \emph{combinable} pairs $((O_w,M_w,D_w),s_w)$ for $w$ as a correct pair for~$w$, where combinable means that $((O_v,M_v,D_v),s_v)$ and $((O_w,M_w,D_w),s_w)$ do not contradict. 
The value $s_w$ results from the shortest distances to the positions involved in the border $v$. 
In~a similar manner we process the cases of two and three children. Finally, we return the smallest score value of a correct state of a border of $c$.

The runtime of our algorithm is linear in the number of vertices of $\polygon$ because each vertex from $T$ can be processed in constant time. This concludes the proof of Theorem~\ref{thm:optimal_for_thin_polyominos}.

\section{Conclusion and future work}\label{sec:conclusion}

We introduced the dispersive art gallery problem and investigated it for vertex guards in polyominoes. 
We developed an algorithm that constructs worst-case optimal solutions of dispersion distance $3$, and showed that it is \NP-complete to decide whether a dispersion distance of $5$ can be achieved. 
We were also able to find a linear-time dynamic programming approach to compute guard sets of maximum dispersion distance for tree-shaped polyominoes.

Several open questions remain. 
Is it possible to close the gap to the worst-case, i.e., is deciding whether a dispersion distance of $4$ can be achieved \NP-hard as well?
How hard is the problem in simple polyominoes? 
Is it possible to compute worst-case solutions in the case of non-simple polyominoes? 
It seems very promising that our methods can be extended to non-simple polyominoes.

What can be said about the ratio between the cardinality of guard sets in optimal solutions for the dispersive and the classic art gallery problem? 
As~shown in~\cref{fig:introduction} this ratio is at least $2$ in simple polyominoes, while the ratio between the dispersion distances increases arbitrarily.

What can be said about the dispersive art gallery problem in terrains, or general polygons?
\bibliography{bibliography}

\end{document}